\documentclass[letterpaper, 10pt, conference]{IEEEtran} 
\normalsize

\DeclareMathAlphabet\mathbfcal{OMS}{cmsy}{b}{n}

\usepackage{amsthm}
\usepackage{bbm}
\usepackage{algorithm}
\usepackage{algorithmic}
\usepackage{graphicx,amsmath,amssymb,latexsym,epsfig,url}
\usepackage{setspace}
\usepackage{psfrag}
\usepackage{array,booktabs,arydshln,xcolor,cite}

\usepackage{lscape}
\newtheorem{theorem}{Theorem}
\newtheorem{proposition}{Proposition}

\newtheorem{corollary}{Corollary}
\newtheorem*{proof of Theorem*}{Proof of Theorem 3}
\newtheorem{proof of Lemma}{Proof of Lemma}
\newtheorem{definition}{Definition}
\newtheorem{lemma}{Lemma}
\newtheorem{remark}{Remark}

\begin{document}

\title{Scheduling Status Updates to Minimize Age of Information with an Energy Harvesting Sensor}

\author{\IEEEauthorblockN{Baran Tan Bacinoglu, Elif Uysal-Biyikoglu}
\IEEEauthorblockA{ METU, Ankara, Turkey.
E-mail:  barantan@metu.edu.tr, uelif@metu.edu.tr,}
}

\bibliographystyle{IEEEtran}

\maketitle

\def\eg{\emph{e.g.}}
\def\ie{\emph{i.e.}}

\begin{abstract}
Age of Information is a measure of the freshness of status updates in monitoring applications and update-based systems. We study a real-time remote sensing scenario with a sensor which is restricted by time-varying energy constraints and battery limitations. The sensor sends updates over a packet erasure channel with no feedback. The problem of finding an age-optimal
threshold policy, with the transmission threshold being a function of the energy state and the estimated current age, is formulated. The average age is analyzed for the unit battery scenario under a memoryless energy arrival process. Somewhat surprisingly, for any finite arrival rate of energy, there is a positive age threshold for transmission, which corresponding to transmitting at a rate lower than that dictated by the rate of energy arrivals. A lower bound on the average age is obtained for general battery size. 

\end{abstract}

\section{Introduction}

The \emph{Age of Information} (AoI), or status age, was proposed \cite{ Kaul2011, Kaul2012} as a quality metric for monitoring applications \cite{ZviedrisESMS10, blueforce} where the freshness of real-time information is crucial. AoI is defined as the amount of time that has elapsed since the most recently received status update or sample was generated at the source. Accordingly, the smaller the age, the closer the information at the receiver to the actual status at the sender.

Recent literature contains analyses of age under various service policies and queuing models \cite{Kaul2011, Kaul2012}, \cite{Ephremides2013, Ephremides2014, Huang2015, Pappas2015, Ephremides2016, Najm2016, DBLP:journals/corr/YatesK16}. A common observation in these studies was that the minimization of age distinctly differs from delay or throughput optimization. Besides these queuing theoretic studies, the transmission scheduling setup in \cite{YinSunInfocom2016} also revealed that throughput and delay optimal update policies can be suboptimal with respect to age. Within the user scheduling formulation in \cite{Kadota2016} for example, average-age optimal scheduling policies turn out to be throughput-optimal, whereas the converse is not always true. 

Optimal transmission of a discrete Markov source observed by an energy harvesting sensor was studied in 
\cite{Nayyar2013} where the optimal transmission strategy of the sensor was described by a threshold policy.

Minimizing age under energy harvesting constraints was considered in \cite{7308962}, which investigated a scenario where an energy harvesting sender device generates and transmits status updates ``at will". Age optimal policies were found under the assumption that the transmission delay for status updates is negligible or deterministic. The model  captured operational scenarios where inter-update intervals are much greater than transmission delay. 

Energy constraints were also considered in \cite{Yates2015}, where the update rate is subject to an upperbound dictated by the average power available to the sensor. In \cite{Yates2015}, in contrast to \cite{7308962}, the randomness is not in the energy arrival process but in the transmission delay or service time. Interestingly, this results in a similar ``lazy" policy for packet transmissions as in \cite{7308962}.  The result was further generalized in \cite{YinSunInfocom2016} showing that the optimal policy will impose a non-zero waiting time in a large class of service time distributions, thus sending updates at a slower rate than that allowed by constraints. A recent study \cite{YinSunISIT2017} similarly showed that, in remote estimation of a Wiener process where samples are forwarded over a channel with random delay, the optimal sampling policy entails a non-zero waiting time between samples even in the absence of a sampling-rate constraint.  

In this paper, we follow the framework in \cite{7308962}
with two major differences: unlike \cite{7308962}, we consider finite energy storage capability. This places a limit on the burstiness of the updating process as updating opportunities (i.e. energy packets) that are not used may be wasted. Secondly, we consider an infinite horizon problem as opposed to a finite time window. As suggested in \cite{7308962}, we consider update policies of threshold-type which are optimal in continuous time for the problem of minimizing average age under stationary energy harvesting
processes.

The contributions of the paper may be summarized as follows: Long term average age achieved by a threshold type policy is analyzed for a packet erasure channel (as in \cite{Huang2015}) where status updates are lost with probability $1-p$. Under a Poisson energy
arrival process, a lower bound for the minimum average age achieved over this channel is expressed based on the minimum average achieved over the lossless
channel ($p = 1$). For a system with a unit-energy buffer, the average age achieved over the lossless channel is derived for Poisson energy arrivals. It is observed that regardless of the energy arrival rate there is a positive age threshold for transmission, such that the average age can be lower that the average age achieved by the zero-wait policy. This is somewhat unexpected as zero-wait is the only energy-conserving and work-conserving policy. Moreover, we show, echoing the result by Yates \cite{Yates2015}, that the age optimal threshold policy sends updates at a slower rate than the upperbound on update rate imposed by the energy arrival process. In particular,  the mean inter-update duration is $\approx 1.307$ times longer than the mean inter-arrival duration of the energy packets.

\section{System Description}
\label{sec:sm}

\begin{figure}[htpb]
    \centering \includegraphics[scale=0.78]{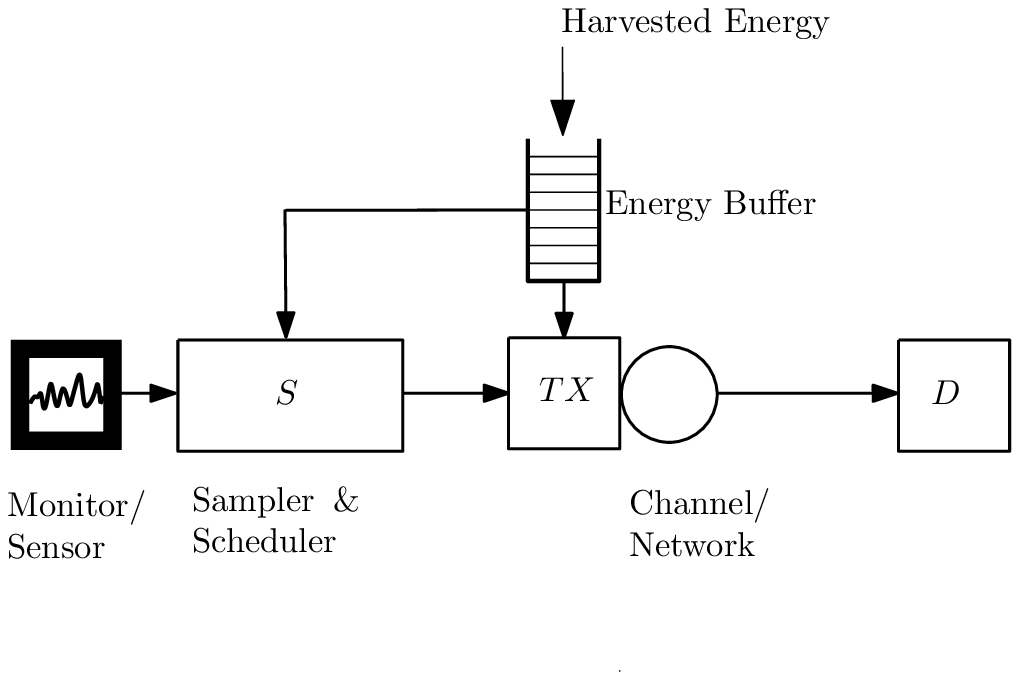}

\caption{ The System Model.}
\label{sensormodelw} 
\end{figure}

Consider an energy harvesting transmitter $TX$ that can send update packets to a destination $D$ through a channel/network but receives no feedback from $D$. The timing of status updates are controlled by a scheduler and sampler $S$
which monitors the energy buffer of $TX$ and the status of a system in real time.  the Age of Information observed at $D$ which is measured as follows:
\[
\Delta (t)= t-U(t)
\]
where $U(t)$ is the time stamp  of the most recent status update packet received by $D$.

Suppose $TX$ has a battery size of $B$ energy units and consumes a unit energy per update transmission where updates are transmitted successfully with probability $p$.
 
Let $N_U(t)$ be the total number of updates transmitted by $TX$ in the time interval$(0,t]$ which is a counting process that can expressed as $N_{U}(t)=\displaystyle\sum_{k} u(t-t_{k})$ where $u(t)$ is the unit step function and $t_{k}$ is the $k^{th}$ update instant.

$S$ is unaware of the age observed at $D$ yet it can compute the expected age $\hat{\Delta}(t)$ at time $t$ based on its previous updates $\lbrace N_U(w)\rbrace_{w \leq t}$, i.e. $\hat{\Delta}(t) =\mathbb {E}[\Delta(t) \mid  \lbrace N_U(w)\rbrace_{w \leq t}]$ assuming $\Delta(0)=0$. The change in $\hat{\Delta}(t) $ can be tracked by the differential equation in below:

\begin{equation}
\dfrac{d}{d t} \hat{\Delta}(t) = 1 - p \hat{\Delta}(t) \dfrac{d}{d t} N_U(t)
\end{equation}

$TX$ harvests energy in order to send its update packets where energy harvests come in unit energy\footnote{While energy packets can refer to actual energy that becomes available to the sender (through energy harvesting, for example), they can also represent transmission opportunities. For example, the energy packets could be replaced with ``tokens" that allow the transmission of update packets in a particular system where update transmission are regulated by ``tokens" either intentionally or due to some physical constraints.}  chunks. Let $N_H(t)$ be the total number of unit energy harvests received by $TX$ in the time interval $(0,t]$ ($N_H(0)=0$) which is a counting process similar to $N_{U}(t)$. Let $E(t)$ be the battery state of $S$ at time $t$, i.e. the number of unit energy harvests stored in $S$. The change in $E(t)$ is governed by the following differential equation:

\begin{equation}
\label{energyupdate}
 \dfrac{d}{d t} E(t) = \mathbbm{1}_{\lbrace E(t)\neq B \rbrace}\dfrac{d}{d t} N_H(t)-\dfrac{d}{d t} N_U(t)
 \end{equation}


\vspace{0.2in}
\section{Minimizing Average Age}
In the problem of minimizing average age, the update policy of $S$ will be taken as a threshold-type policy as suggested by Theorem 3 in  \cite{7308962}. 
Accordingly, for a given threshold function $\displaystyle \tau(\cdot): \lbrace 0, 1,2,.....,B\rbrace\rightarrow [0,+\infty)$ such that  $\tau(0)=+\infty$, the update policy satisfies $N_U(0)=0$ and:
\begin{equation}
\label{onlinepolicy}
\dfrac{d}{d t} N_U(t) =\mathbbm{1}_{ \lbrace \hat{\Delta}(t) \geq \tau(E(t))  \rbrace }\delta(t)
\end{equation}
where $\delta(t)$ is the Dirac delta function.

Note that the vector of thresholds given by $[\tau_{1},\tau_{2},......, \tau_{B}]$ for which $\tau_{m}=\tau(m)$ also identifies an update policy.

We consider the minimization of the average age $\bar{\Delta}$:
\begin{equation}
\bar{\Delta}=\displaystyle\lim\sup_{T \rightarrow \infty}\frac{\mathbb {E}\left[ \displaystyle\int_{0}^{T}\Delta(t)  dt \right] }{T}
\end{equation}
for update policies with threshold constraints as defined in (\ref{onlinepolicy}).

As $\hat{\Delta}(t) =\mathbb {E}[\Delta(t) \mid  \lbrace N_U(w)\rbrace_{w \leq t}]$, the average age can be also expressed as follows:

\begin{equation}
\label{aveage}
\bar{\Delta}=\displaystyle\lim\sup_{T \rightarrow \infty}\frac{\mathbb {E}\left[ \displaystyle\int_{0}^{T}\hat{\Delta}(t)  dt \right] }{T}
\end{equation}

\begin{definition}
An optimal threshold policy is defined as a threshold policy that obeys (\ref{onlinepolicy}) and minimizes (\ref{aveage}).
\end{definition}

\begin{remark}
In any optimal threshold policy, $\tau(0)=+\infty$ and $\tau(m)\rightarrow 0$ as $m=+\infty$.
\end{remark}
\begin{proof}
Clearly,  $\tau(0)=+\infty$ as $S$ cannot send an update without energy. At the other extreme of infinite stored energy, sending an update does not affect the energy state while only lowering age, hence it is optimal to send an update at any positive value of age.
\end{proof}

\begin{definition}
A zero-wait policy is defined as a policy that obeys (\ref{onlinepolicy}) with $\tau(m)=0$ for all $1\leq m \leq B$ and $\tau(0)=+\infty$.
\end{definition}

The average age achieved by any update policy is characterized by the inter-update durations and the expected cumulative age between updates. Let $X_{k}=t_{k+1}-t_{k}$ be the inter-update duration and $c_{k}$ be the expected cumulative age , i.e. $ c_{k}= \mathbb {E}\left[ \displaystyle\int_{t_{k}}^{t_{k+1}}\hat{\Delta}(t)  dt \right]
$, between the updates at $t_{k}$ and $t_{k+1}$. Define $\bar{X}_{N}$  as $\bar{X}_{N}=\frac{1}{N}\displaystyle\sum_{k=0}^{N} X_{k}$ and $\bar{C}_{N}$ as $\bar{C}_{N}=\frac{1}{N}\displaystyle\sum_{k=0}^{N} c_{k}$, then according to the below lemma, the average age can be computed as the ratio of the limits $\bar{X}_{N}$ and $\bar{C}_{N}$.

\begin{lemma}
\label{averageagelim}
For all policies and energy arrival processes for which the limits in (\ref{averageX}) and (\ref{averageC})  exist, the average age can be computed as $\bar{\Delta}=\bar{C}/\bar{X}$.

\begin{equation}
\label{averageX}
\displaystyle\lim_{N\rightarrow \infty} \bar{X}_{N} =\bar{X}, a.s.
\end{equation}
where $\displaystyle\lim_{N\rightarrow \infty} \bar{X}_{N}$ is evaluated for all sample realizations except a set of probability zero.
\begin{equation}
\label{averageC}
\displaystyle\lim_{N\rightarrow \infty}\bar{C}_{N} = \bar{C}
\end{equation}

\end{lemma}

\begin{proof}
Consider the inequalities below:

%
%

\[
\displaystyle\lim_{T\rightarrow \infty}\frac{\displaystyle\sum_{k=0}^{N_U(T)}c_{k} }{T} \leq \bar{\Delta} \leq \displaystyle\lim_{T\rightarrow \infty}\frac{\displaystyle\sum_{k=0}^{N_U(T)+1}c_{k} }{T}
\]

where 

\[
\displaystyle\lim_{T\rightarrow \infty}\frac{\displaystyle\sum_{k=0}^{N_U(T)}c_{k} }{T} =\displaystyle\lim_{N\rightarrow \infty}\frac{\displaystyle\sum_{k=0}^{N}c_{k}/N }{\displaystyle\sum_{k=0}^{N} X_{k}/N}=\frac{\bar{C}}{\bar{X}}
\]
and similarly $
\displaystyle\lim_{T\rightarrow \infty}\frac{\displaystyle\sum_{k=0}^{N_U(T)+1}c_{k} }{T}=\frac{\bar{C}}{\bar{X}}
$. Therefore, $\bar{\Delta}=\bar{C}/\bar{X}$.
\end{proof}

\subsection{Poisson Arrivals}
Under Poisson energy arrivals, we investigate policies in the form\footnote{The policies in the form (\ref{onlinepolicy}) are particularly suitable for memoryless energy arrivals as these  policies are oblivious to the history of energy arrivals.}  (\ref{onlinepolicy}) that minimize $\bar{\Delta}$.
\begin{proposition}
When $N_H(t)$ 
 arrival process with rate $\mu_{H}$, 

(i) The average age achieved by zero-wait policy is $\frac{1}{p\mu_{H}}$ .

(ii) The arrival process at $D$ is Poisson process with rate $p\mu_{H}$ .
\end{proposition}
\begin{proof}
For zero-wait policy, the inter-update durations $X_{k}$s are identical to the inter-arrival durations of the energy arrival process. Accordingly, under the assumption that $N_H(t)$ is a Poisson arrival process, $X_{k}$s are independently and exponentially distributed with rate $\mu_{H}$. Therefore, 
\[
\displaystyle\lim_{N\rightarrow \infty} \bar{X}_{N} = \frac{1}{\mu_{H}}
\]
which means the limit in (\ref{averageX}) of Lemma \ref{averageagelim} exists.

Consider $
c_{k}=\mathbb {E}\left[ \displaystyle\int_{t_{k}}^{t_{k+1}}\hat{\Delta}(t)  dt \right]
$
which can be computed as:
\[
c_{k}=\frac{1}{2}\mathbb {E}\left[X_{k}^{2}\right]+\mathbb {E}\left[\hat{\Delta}(t_{k})X_{k}\right]
\]
As $X_{k}$ is independent from $\hat{\Delta}(t_{k})$ for the zero-wait policy, $\mathbb {E}\left[\hat{\Delta}(t_{k})X_{k}\right]=\frac{1}{\mu_{H}}\mathbb{E}\left[\hat{\Delta}(t_{k})\right]$, hence:
\[
c_{k}=\frac{1}{\mu_{H}^{2}}+\frac{1}{\mu_{H}}\mathbb{E}\left[\hat{\Delta}(t_{k})\right]
\]

As $\hat{\Delta}(t_{k+1})=(1-p)(\hat{\Delta}(t_{k})+X_{k})$:

\[
\displaystyle\lim_{k\rightarrow \infty}\mathbb{E}\left[\hat{\Delta}(t_{k})\right]=(\frac{1-p}{p})\frac{1}{\mu_{H}}
\]
which means the limit in (\ref{averageC}) of Lemma \ref{averageagelim} exists and $\bar{C}=\frac{1}{\mu_{H}^{2}}+(\frac{1-p}{p})\frac{1}{\mu_{H}^{2}}$. Therefore,
\[
\bar{\Delta}=\frac{1}{\mu_{H}}+(\frac{1-p}{p})\frac{1}{\mu_{H}} =\frac{1}{p\mu_{H}}
\]
Note that the arrival process at $D$ is a splitting of the update process with probability $p$, thus it is a Poisson process with rate $p\mu_{H}$. 
\end{proof}
\begin{remark}
Under Poisson energy arrivals, there is always a threshold policy that obeys (\ref{onlinepolicy}) and achieves a finite average age such that $
\displaystyle\min_{\lbrace \tau_{1},\tau_{2},....,\tau_{B}\rbrace \in [0,+\infty)^{B}}\bar{\Delta} \leq \frac{1}{p\mu_{H}}
$.
 
\end{remark}

\begin{proposition}
\label{finitemoments} 
When $N_H(t)$ is a Poisson arrival process with rate $\mu_{H}$, for every policy that obeys (\ref{onlinepolicy}), $\mathbb {E}[X_{k}]$ and $\mathbb {E}[X_{k}^{2}]$ are upper bounded as in the following inequalities:
\[
\mathbb {E}[X_{k}] \leq  \tau_{max}+\frac{1}{\mu_{H}}e^{-\mu_{H}\tau_{max}}
\]
\[
\mathbb {E}[X_{k}^{2}] \leq  \tau_{max}^{2}+(\frac{2}{\mu_{H}^{2}}+\frac{2}{\mu_{H}}\tau_{max})e^{-\mu_{H}\tau_{max}}
\]
where $\tau_{max}=\displaystyle\max \lbrace \tau_{1},\tau_{2},....,\tau_{B}\rbrace$
\end{proposition}
\begin{proof}
The expectations of $\mathbb {E}[X_{k}]$ and $\mathbb {E}[X_{k}^{2}]$ can be upper bounded by the law of total expectations as follows:
\[
\mathbb {E}[X_{k}]= \mathbb {E}[\mathbb {E}[X_{k}\mid \hat{\Delta}(t_{k}),E(t_{k})]]
\]
\[
\leq \displaystyle\max_{\Delta,m}\mathbb {E}[X_{k}\mid \hat{\Delta}(t_{k})=\Delta,E(t_{k})=m]
\]
\[
\mathbb {E}[X_{k}^{2}]= \mathbb {E}[\mathbb {E}[X_{k}^{2}\mid \hat{\Delta}(t_{k}),E(t_{k})]]
\]
\[
\leq \displaystyle\max_{\Delta,m}\mathbb {E}[X_{k}^{2}\mid \hat{\Delta}(t_{k})=\Delta,E(t_{k})=m]
\]
Let $R$ be the residual time (exponentially distributed with mean $\frac{1}{\mu_{H}}$ )  to the next energy arrival for any given time. Then, $\mathbb {E}[X_{k}\mid \hat{\Delta}(t_{k})=\Delta,E(t_{k})=m]$ and $\mathbb {E}[X_{k}^{2}\mid \hat{\Delta}(t_{k})=\Delta,E(t_{k})=m]$ can be upper bounded as  in below:

(i) For $m=0$:
\[
\mathbb {E}[X_{k}\mid \hat{\Delta}(t_{k})=\Delta,E(t_{k})=0]\leq 
\]
\[
(\tau_{max}-\Delta)\Pr(R \leq \tau_{1}-\Delta)+\underbrace{\mathbb {E}[R \mid R > \tau_{1}-\Delta ]}_{\mathbb {E}[R]+[\tau_{1}-\Delta]^{+}}\Pr(R > \tau_{1}-\Delta)
\]
\[
\mathbb {E}[X_{k}^{2}\mid \hat{\Delta}(t_{k})=\Delta,E(t_{k})=0]\leq 
\]
\[
(\tau_{max}-\Delta)^{2}\!\Pr(R \leq \tau_{1}-\Delta)+\underbrace{\mathbb {E}[R^{2} \!\mid\! R > \tau_{1}-\Delta ]}_{\mathbb {E}[(R+[\tau_{1}-\Delta]^{+})^{2}]}\Pr(R > \tau_{1}-\Delta)
\]
(i-a) For $\Delta \geq \tau_{1}$, 
\[
\mathbb {E}[X_{k}\mid \hat{\Delta}(t_{k})=\Delta \geq \tau_{m},E(t_{k})=0]=\mathbb {E}[R]=\frac{1}{\mu_{H}}
\]
\[
\mathbb {E}[X_{k}^{2}\mid \hat{\Delta}(t_{k})=\Delta \geq \tau_{m},E(t_{k})=0]=\mathbb {E}[R^{2}]=\frac{2}{\mu_{H}^{2}}
\]
(i-b) For $\Delta < \tau_{1}$,
\[
\mathbb {E}[X_{k}\mid \hat{\Delta}(t_{k})=\Delta <  \tau_{m},E(t_{k})=0]\leq
\]
\[
(\tau_{max}-\Delta)(1-e^{-\mu_{H}(\tau_{1}-\Delta)})+(\tau_{1}-\Delta +\frac{1}{\mu_{H}})e^{-\mu_{H}(\tau_{1}-\Delta)}
\]
\[
\mathbb {E}[X_{k}^{2}\mid \hat{\Delta}(t_{k})=\Delta <  \tau_{m},E(t_{k})=0]\leq
\]
\[
(\tau_{max}-\Delta)^{2}(1-e^{-\mu_{H}(\tau_{1}-\Delta)})
\]
\[
+(\frac{2}{\mu_{H}^{2}}+\frac{2}{\mu_{H}}(\tau_{1}-\Delta) +(\tau_{1}-\Delta)^{2})e^{-\mu_{H}(\tau_{1}-\Delta)}
\]
As both of the upper bounds are nondecreasing with $\tau_{1}-\Delta$, they are maximized when $\tau_{1}-\Delta$ is maximized where $\Delta=0$ and $\tau_{1}= \tau_{max}$. Therefore,
\[
\mathbb {E}[X_{k}\mid \hat{\Delta}(t_{k})=\Delta,E(t_{k})=0]\leq
\]
\[
\tau_{max}(1-e^{-\mu_{H}\tau_{max}})+(\tau_{max}+\frac{1}{\mu_{H}})e^{-\mu_{H}\tau_{max}}
\]
\[
\mathbb {E}[X_{k}^{2}\mid \hat{\Delta}(t_{k})=\Delta,E(t_{k})=0]\leq
\]
\[
\tau_{max}^{2}(1-e^{-\mu_{H}\tau_{max}})+(\frac{2}{\mu_{H}^{2}}+\frac{2}{\mu_{H}}\tau_{max} +\tau_{max}^{2})e^{-\mu_{H}\tau_{max}}
\]
(ii) For $m>0$:
\[
\mathbb {E}[X_{k}\mid \hat{\Delta}(t_{k})=\Delta,E(t_{k})=m]\leq \tau_{max}-\Delta
\]
\[
\leq \tau_{max}(1-e^{-\mu_{H}\tau_{max}})+(\tau_{max}+\frac{1}{\mu_{H}})e^{-\mu_{H}\tau_{max}}
\]
\[
\mathbb {E}[X_{k}^{2}\mid \hat{\Delta}(t_{k})=\Delta,E(t_{k})=m]\leq (\tau_{max}-\Delta)^{2}
\]
\[
\leq \tau_{max}^{2}(1-e^{-\mu_{H}\tau_{max}})+(\frac{2}{\mu_{H}^{2}}+\frac{2}{\mu_{H}}\tau_{max} +\tau_{max}^{2})e^{-\mu_{H}\tau_{max}}
\]

\end{proof}

\begin{definition}
Define $\bar{\Delta}_{B}^{*}(\mu_{H},p)$ as the the minimum average age achievable by any policy that obeys (\ref{onlinepolicy}) for an $S$ that transmits update packets with success probability $p$ and has an energy buffer of $B$ units which receives Poisson energy arrivals with rate $\mu_{H}$. 
\end{definition}

\begin{lemma}
\label{poissonbound} 
Under Poisson energy arrivals, the minimum average age is lower bounded by the minimum average age achieved over the lossless channel as in the following:

\begin{equation}
\bar{\Delta}_{B}^{*}(p\mu_{H},1)\leq \bar{\Delta}_{B}^{*}(\mu_{H},p)
\end{equation}  

\end{lemma}

\begin{proof}

Note that the system with $p<1$ corresponds to each energy packet that is used to send an update, being wasted due to an erasure with probability $p$, independently of the battery state.  Let there be two types of energy packets: type-1 energy packets that, when consumed, cause the packet to be lost, and type-2 energy packets that, when consumed result in successful transmission. The success process is independent of the battery state, and each packet is independently type-1 or type-2 with probabilities $p$ and $1-p$. Suppose energy packets are stored in two different internal buffers of the sender corresponding their type, i.e. type-1 and type-2 buffers. When a packet is lost, an energy packet is used from the type-1 buffer, so on. The capacity of type-2/type-1 buffers are also dynamically determined by the number of type-1/type-2 packets as the total capacity of these buffers is $B$. If $S$ cannot differentiate the type of recieved energy packets, then this interpretation does not change the system operation. Now, suppose $S$ can differentiate  type-1 and type-2 energy packets, so it simply rejects type-1 energy packets while storing others. Clearly, the age can be reduced in this case as $S$ no longer wastes its buffer capacity and serving time on type-1 buffer which does not produce succesfullly transmitted packets. Note that by independent splitting, this is equivalent to having Poisson energy arrivals with rate $p\mu_{H}$ and success probability $1$.    
\end{proof}

According to the above result, for a particular buffer limit $B$, the minimum average age found assuming $p=1$  gives a lower bound on the minimum average for  $p<1$ case. The case for $p=1$ is more tractable as the age is always reset to zero after an update. In this case, inter-update durations are generated by a Markov renewal process depending only on energy states rather than the age itself.  

Next, we consider the scenario where $B=1$ and derive the expression for the average age in the following theorem. 

\begin{theorem}
\label{p1age}
When $p=1$ and $B=1$, the average age $\bar{\Delta}$ can be expressed as follows:

\begin{equation}
\bar{\Delta}=\frac{1}{2}\frac{\tau_{1}^{2}+(\frac{2}{\mu_{H}^{2}}+\frac{2}{\mu_{H}}\tau_{1})e^{-\mu_{H}\tau_{1}}}{\tau_{1}+\frac{1}{\mu_{H}}e^{-\mu_{H}\tau_{1}}}
\end{equation}

and $\bar{\Delta}$ is minimized to $\frac{2W(\frac{1}{\sqrt{2}})}{\mu_{H}}$, i.e. $\bar{\Delta}_{1}^{*}(\mu_{H},1)=\frac{2W(\frac{1}{\sqrt{2}})}{\mu_{H}}$ when $\tau_{1}=\frac{2W(\frac{1}{\sqrt{2}})}{\mu_{H}}$ where $W(\cdot)$ is the Lambert-W function.

\end{theorem}
\begin{proof}
As $p=1$, the expected age just after an update instant is always reset to zero and the policies for $B=1$ can be described by only one threshold which is $\tau_{1}$. In particular, the inter-update durations $X_{k}$s are i.i.d. which means $\mathbb {E}[X_{0}]=\mathbb {E}[X_{1}]=....=\mathbb {E}[X_{k}]$ and $\mathbb {E}[X_{0}^{2}]=\mathbb {E}[X_{1}^{2}]=....=\mathbb {E}[X_{k}^{2}]$. The limits in  (\ref{averageX}) and (\ref{averageC}) of Lemma \ref{averageagelim} exist as $X_{k}$s are i.i.d. and $c_{k}=\frac{1}{2}\mathbb {E}[X_{k}^{2}]$ is finite.

Accordingly, the average age $\bar{\Delta}$ is equal to $\frac{\mathbb {E}[X_{k}^{2}]}{2\mathbb {E}[X_{k}]}$ for any $k$. The values of $\mathbb {E}[X_{k}^{2}]$ and $\mathbb {E}[X_{k}]$ for a given $\tau_{1}$ can be computed by the law of total expectations as follows:
\[
\mathbb {E}[X_{k}^{2}]=\tau_{1}^{2}(1-e^{-\mu_{H}\tau_{1}})+(\frac{2}{\mu_{H}^{2}}+\frac{2}{\mu_{H}}\tau_{1} +\tau_{1}^{2})e^{-\mu_{H}\tau_{1}}
\]
\[
\mathbb {E}[X_{k}]=\tau_{1}(1-e^{-\mu_{H}\tau_{1}})+(\tau_{1}+\frac{1}{\mu_{H}})e^{-\mu_{H}\tau_{1}}
\]
Therefore,
\[
\bar{\Delta}=\frac{1}{2}\frac{\tau_{1}^{2}(1-e^{-\mu_{H}\tau_{1}})+(\frac{2}{\mu_{H}^{2}}+\frac{2}{\mu_{H}}\tau_{1} +\tau_{1}^{2})e^{-\mu_{H}\tau_{1}}}{\tau_{1}(1-e^{-\mu_{H}\tau_{1}})+(\tau_{1}+\frac{1}{\mu_{H}})e^{-\mu_{H}\tau_{1}}}
\]
Consider the ratio $\bar{\Delta}/\bar{\Delta}_{min}$ where $\bar{\Delta}_{min}=\displaystyle\min_{\tau_{1}\geq 0}\bar{\Delta}(\tau_{1})$.
As $\frac{\bar{\Delta}}{\bar{\Delta}_{min}} \geq 1$, $
\displaystyle\arg\min_{\tau_{1}\geq 0}\bar{\Delta}(\tau_{1})= \displaystyle\arg\min_{\tau_{1}\geq 0} J(\tau_{1})
$ 
where $\displaystyle\arg\min_{\tau_{1}\geq 0} J(\tau_{1})=0$ and
\[
J(\tau_{1})= \tau_{1}^{2}+(\frac{2}{\mu_{H}^{2}}+\frac{2}{\mu_{H}}\tau_{1} )e^{-\mu_{H}\tau_{1}}-\bar{\Delta}_{min}\left[ 2\tau_{1}+\frac{2}{\mu_{H}}e^{-\mu_{H}\tau_{1}}\right] 
\]
When $\dfrac{d}{d \tau_{1}}J(\tau_{1})=0$, $(\tau_{1}-\bar{\Delta}_{min})(1-e^{-\mu_{H}\tau_{1}})=0$
which means  $
\displaystyle\arg\min_{\tau_{1}\geq 0} J(\tau_{1})=\bar{\Delta}_{min}
$
since $(1-e^{-\mu_{H}\tau_{1}})>0$. Hence, $\bar{\Delta}_{min}$ is the fixed point of $\bar{\Delta}(\tau_{1})$ which satisfies $\bar{\Delta}_{min}^{2}=\frac{2}{\mu_{H}}e^{-\mu_{H}\bar{\Delta}_{min}}$. Therefore $
\bar{\Delta}_{min}=\frac{2W(\frac{1}{\sqrt{2}})}{\mu_{H}}
$.
\end{proof}

\begin{corollary}
For $B=1$ and $p=1$, as the optimal value of the threshold $\tau_{1}$ is $\frac{2W(\frac{1}{\sqrt{2}})}{\mu_{H}}$, the mean inter-update duration $\mathbb {E}[X_{k}]=\tau_{1}+\frac{1}{\mu_{H}}e^{-\mu_{H}\tau_{1}}$ is $2W(\frac{1}{\sqrt{2}})+e^{-2W(\frac{1}{\sqrt{2}})}$ times larger than $\frac{1}{\mu_{H}}$. As $W(\cdot)$ is the Lambert-W function, $2W(\frac{1}{\sqrt{2}})+e^{-2W(\frac{1}{\sqrt{2}})}$ is approximately $1.307$.
\end{corollary}

\begin{figure}[htpb]
\centering
  \begin{psfrags}
    \psfrag{X}[t]{$\tau_{1}$}
    \psfrag{Y}[b]{Average Age  $\bar{\Delta}$}
    \psfrag{A}[l]{$\mu_{H}=10$}
    \psfrag{B}[l]{$\mu_{H}=15$}
    \psfrag{C}[l]{$\mu_{H}=20$}
    \psfrag{D}[l]{$\mu_{H}=25$}
    \psfrag{E}[l]{$\mu_{H}=30$}
    \includegraphics[scale=0.24]{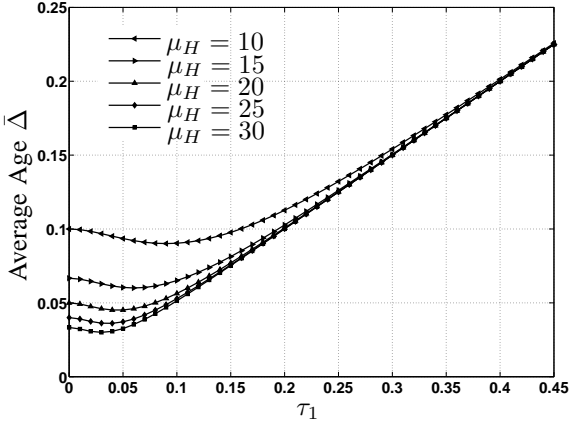}
    \end{psfrags}
\caption{Average age $\bar{\Delta}$ versus $\tau_{1}$ when $B=1$ and $p=1$.}
\label{fig:maxthroughputcomppowerhalv}
\end{figure}


A consequence of the above theorem is that $\bar{\Delta}_{1}^{*}(\mu_{H},p)\geq \frac{2W(\frac{1}{\sqrt{2}})}{p\mu_{H}}$ which means, with a unit buffer, a threshold policy can achieve up to $10\%$ ($2W(\frac{1}{\sqrt{2}})\approx 0.901$) lower average age than that of the zero-wait policy.

In the below theorem, we show a more general lower bound on $\bar{\Delta}_{B}^{*}(\mu_{H},p)$ for any buffer limit $B\geq 1$.

\begin{theorem}
\label{lowestage}
When the energy arrival process is Poisson, $\bar{\Delta}_{B}^{*}(\mu_{H},p)\geq \frac{1}{2p\mu_{H}}$ for any $B$ and $p$.
\end{theorem}

\begin{proof}
First, consider the case for $p=1$ where $c_{k}=\frac{1}{2}\mathbb {E}[X_{k}^{2}]$ and $\mathbb {E}[X_{k}]$ are finite by Proposition \ref{finitemoments}. This means that the average age can be expressed as follows by Lemma \ref{averageagelim}:

\[
\bar{\Delta}=\frac{\lim_{N\rightarrow +\infty}\frac{1}{2N}\sum_{k=0}^{N}\mathbb {E}[X_{k}^{2}]}{\lim_{N\rightarrow +\infty}\frac{1}{N}\sum_{k=0}^{N}\mathbb {E}[X_{k}]}
\] 

As $\mathbb {E}[X_{k}^{2}] \geq \mathbb {E}[X_{k}]$:

\[
\frac{\frac{1}{2N}\sum_{k=0}^{N}\mathbb {E}[X_{k}^{2}]}{\frac{1}{N}\sum_{k=0}^{N}\mathbb {E}[X_{k}]} \geq \frac{\frac{1}{2N}\sum_{k=0}^{N}\mathbb {E}[X_{k}]^{2}}{\frac{1}{N}\sum_{k=0}^{N}\mathbb {E}[X_{k}]}
\]
By Jensen's inequality, $\frac{1}{N}\sum_{k=0}^{N}\mathbb {E}[X_{k}]^{2}\geq \left( \frac{1}{N}\sum_{k=0}^{N}\mathbb {E}[X_{k}]\right)^{2} $ hence $
\bar{\Delta}\geq\lim_{N\rightarrow +\infty}
\frac{1}{2N}\sum_{k=0}^{N}\mathbb {E}[X_{k}]
$
which means $\bar{\Delta} \geq \frac{1}{2\mu_{H}}$ as $\lim_{N\rightarrow +\infty}
\frac{1}{N}\sum_{k=0}^{N}\mathbb {E}[X_{k}]$ cannot be lower than $\frac{1}{\mu_{H}}$ for a given Poisson energy arrival process with rate $\mu_{H}$. Accordingly, $\bar{\Delta}_{B}^{*}(\mu_{H},1)=\frac{1}{2\mu_{H}}+f_{B}^{*}(\mu_{H})
$
where $f_{B}^{*}(\mu_{H})$ is a nonnegative function of $\mu_{H}$. Therefore, by Lemma  \ref{poissonbound}, $\bar{\Delta}_{B}^{*}(\mu_{H},p)\geq\frac{1}{2p\mu_{H}}+f_{B}^{*}(p\mu_{H})$.

\end{proof}

\begin{theorem}
\label{equidistantupdates}
When $B=\infty$ and $p=1$, the threshold policy that obeys (\ref{onlinepolicy}) and has thresholds $\tau =\tau_{m}=\frac{1}{\mu_{H}}+\epsilon$ for every $m>0$ and an arbitrary $\epsilon>0$, achieves the average age $\bar{\Delta}=\frac{1}{2}(\frac{1}{\mu_{H}}+\epsilon)$ as the energy arrival process is Poisson with rate $\mu_{H}$.
\end{theorem}

\begin{proof}

Consider the case where $E(t_{k-1})\geq 2$ for some $k>1$. If $E(t_{k-1})\geq 2$, it is guaranteed that $\hat{\Delta}(t_{k})=0$ and $X_{k}=\frac{1}{\mu_{H}}+\epsilon$. Accordingly,

\[
\mathbb {E}\left[ \displaystyle\int_{t_{k}}^{t_{k+1}}\hat{\Delta}(t)  dt \mid E(t_{k-1})\geq 2\right] = \frac{1}{2}(\frac{1}{\mu_{H}}+\epsilon)^{2}
\]
\[
\mathbb {E}\left[X_{k}\mid E(t_{k-1})\geq 2\right] = \frac{1}{\mu_{H}}+\epsilon
\]
Now consider the energy state $E(t)$ at an arbitrary time $t$ which satisfies $E(t)=N_{H}(t)-N_{U}(t)$ from Eq.\ref{energyupdate}. As $\hat{\Delta}(t)=0$ after every update, age increases linearly toward the threshold $\tau = \frac{1}{\mu_{H}}+\epsilon$, inter-update durations cannot be shorter than $\frac{1}{\mu_{H}}+\epsilon$ So $N_{U}(t)\leq t/(\frac{1}{\mu_{H}}+\epsilon)$ and $E(t) \geq N_{H}(t)-t/(\frac{1}{\mu_{H}}+\epsilon)$. The energy harvested until time $t$ is Poisson distributed with mean $t\mu_{H}$ which means $\Pr(\displaystyle\lim_{t\rightarrow +\infty}N_{H}(t)/t=\mu_{H})=1 $, accordingly
$\Pr(\displaystyle\lim_{t\rightarrow +\infty}E(t)/t \geq \mu_{H}-\frac{\mu_{H}}{1+\mu_{H}\epsilon})=1$ thus for any $M>0$ $\displaystyle\lim_{t\rightarrow +\infty}\Pr(E(t) \leq M)=0$. This shows $\displaystyle\lim_{k\rightarrow +\infty}\Pr(E(t_{k}) \leq M)=0$ (**) for any finite $M$. As,

\[
c_{k}=\mathbb {E}\left[ \displaystyle\int_{t_{k}}^{t_{k+1}}\hat{\Delta}(t)  dt \mid E(t_{k-1})\geq 2\right]\Pr(E(t_{k-1})\geq 2)
\]
\[
+\mathbb {E}\left[ \displaystyle\int_{t_{k}}^{t_{k+1}}\hat{\Delta}(t)  dt \mid E(t_{k-1})< 2\right]\Pr(E(t_{k-1})< 2) 
\](*)

As $k\rightarrow +\infty$, the second term in the RHS of (*) vanishes by (**) hence we obtain
$\displaystyle\lim_{k\rightarrow +\infty} c_{k} = \frac{1}{2}(\frac{1}{\mu_{H}}+\epsilon)^{2}$ similarly,
$\displaystyle\lim_{k\rightarrow +\infty} \mathbb {E}[X_{k}] = \frac{1}{\mu_{H}}+\epsilon$.

Therefore, by Lemma \ref{averageagelim},

\[
\bar{\Delta} = \frac{\displaystyle\lim_{N\rightarrow +\infty}\frac{1}{N}\sum_{k=0}^{N}c_{k}}{\displaystyle\lim_{N\rightarrow +\infty}\frac{1}{N}\sum_{k=0}^{N}\mathbb {E}[X_{k}] }=\frac{1}{2}(\frac{1}{\mu_{H}}+\epsilon)
\]

\end{proof}

Note that the policy in Theorem \ref{equidistantupdates} achieves the lower limit on the average age in Theorem \ref{lowestage}, for $p=1$ when $\epsilon$ is arbitrarily small.

\bibliography{AgeOfInformation}

\begin{thebibliography}{10}
\providecommand{\url}[1]{#1}
\csname url@samestyle\endcsname
\providecommand{\newblock}{\relax}
\providecommand{\bibinfo}[2]{#2}
\providecommand{\BIBentrySTDinterwordspacing}{\spaceskip=0pt\relax}
\providecommand{\BIBentryALTinterwordstretchfactor}{4}
\providecommand{\BIBentryALTinterwordspacing}{\spaceskip=\fontdimen2\font plus
\BIBentryALTinterwordstretchfactor\fontdimen3\font minus
  \fontdimen4\font\relax}
\providecommand{\BIBforeignlanguage}[2]{{%
\expandafter\ifx\csname l@#1\endcsname\relax
\typeout{** WARNING: IEEEtran.bst: No hyphenation pattern has been}%
\typeout{** loaded for the language `#1'. Using the pattern for}%
\typeout{** the default language instead.}%
\else
\language=\csname l@#1\endcsname
\fi
#2}}
\providecommand{\BIBdecl}{\relax}
\BIBdecl

\bibitem{Kaul2011}
S.~Kaul, M.~Gruteser, V.~Rai, and J.~Kenney, ``Minimizing age of information in
  vehicular networks,'' in \emph{Sensor, Mesh and Ad Hoc Communications and
  Networks (SECON), 2011 8th Annual IEEE Communications Society Conference on},
  June 2011, pp. 350--358.

\bibitem{Kaul2012}
S.~Kaul, R.~Yates, and M.~Gruteser, ``Real-time status: How often should one
  update?'' in \emph{INFOCOM 2012}, pp. 2731--2735.

\bibitem{ZviedrisESMS10}
R.~Zviedris, A.~Elsts, G.~Strazdins, A.~Mednis, and L.~Selavo, ``Lynxnet: Wild
  animal monitoring using sensor networks,'' in \emph{REALWSN 2010}, 2010, pp.
  170--173.

\bibitem{blueforce}
K.~R. Chevli, P.~Kim, A.~Kagel, D.~Moy, R.~Pattay, R.~Nichols, and A.~D.
  Goldfinger, ``Blue force tracking network modeling and simulation,'' in
  \emph{MILCOM 2006}, Oct 2006, pp. 1--7.

\bibitem{Ephremides2013}
C.~Kam, S.~Kompella, and A.~Ephremides, ``Age of information under random
  updates,'' in \emph{IEEE ISIT}, July 2013, pp. 66--70.

\bibitem{Ephremides2014}
M.~Costa, M.~Codreanu, and A.~Ephremides, ``Age of information with packet
  management,'' in \emph{IEEE ISIT}, June 2014, pp. 1583--1587.

\bibitem{Huang2015}
L.~Huang and E.~Modiano, ``Optimizing age-of-information in a multi-class
  queueing system,'' in \emph{IEEE ISIT}, June 2015, pp. 1681--1685.

\bibitem{Pappas2015}
N.~Pappas, J.~Gunnarsson, L.~Kratz, M.~Kountouris, and V.~Angelakis, ``Age of
  information of multiple sources with queue management,'' in \emph{2015 ICC},
  June 2015, pp. 5935--5940.

\bibitem{Ephremides2016}
C.~Kam, S.~Kompella, G.~D. Nguyen, and A.~Ephremides, ``Effect of message
  transmission path diversity on status age,'' \emph{IEEE Transactions on
  Information Theory}, vol.~62, no.~3, pp. 1360--1374, March 2016.

\bibitem{Najm2016}
E.~Najm and R.~Nasser, ``Age of information: The gamma awakening,'' in
  \emph{IEEE ISIT}, July 2016, pp. 2574--2578.

\bibitem{DBLP:journals/corr/YatesK16}
\BIBentryALTinterwordspacing
R.~D. Yates and S.~K. Kaul, ``The age of information: Real-time status updating
  by multiple sources,'' \emph{CoRR}, vol. abs/1608.08622, 2016. [Online].
  Available: \url{http://arxiv.org/abs/1608.08622}
\BIBentrySTDinterwordspacing

\bibitem{YinSunInfocom2016}
Y.~Sun, E.~Uysal-Biyikoglu, R.~Yates, C.~E. Koksal, and N.~B. Shroff, ``Update
  or wait: How to keep your data fresh,'' in \emph{IEEE INFOCOM 2016}, April
  2016, pp. 1--9.

\bibitem{Kadota2016}
R.~S. I.~Kadota, E. Uysal-Biyikoglu and E.~Modiano, ``Rminimizing age of
  information in broadcast wireless networks,'' in \emph{54th Annual Allerton
  Conf.On on Communi­cation, Control, and Computing}, September 2016.

\bibitem{Nayyar2013}
D.~T. A.~Nayyar, T.~Basar and V.~V. Veeravalli, ``Optimal strategies for
  communication and remote estimation with an energy harvesting sensor,''
  \emph{{IEEE} Transactions on Automatic Control}, vol.~58, no.~9, 2013.

\bibitem{7308962}
B.~T. Bacinoglu, E.~T. Ceran, and E.~Uysal-Biyikoglu, ``Age of information
  under energy replenishment constraints,'' in \emph{2015 ITA}, Feb 2015, pp.
  25--31.

\bibitem{Yates2015}
R.~D. Yates, ``Lazy is timely: Status updates by an energy harvesting source,''
  in \emph{IEEE ISIT}, June 2015, pp. 3008--3012.

\bibitem{YinSunISIT2017}
\BIBentryALTinterwordspacing
Y.~Sun, Y.~Polyanskiy, and E.~Uysal-Biyikoglu, ``Age-of-information and remote
  estimation over a channel with random delay,'' 2017. [Online]. Available:
  \url{http://www2.ece.ohio-state.edu/~suny/paper/conference/sampling_Wiener_c%
onf.pdf}
\BIBentrySTDinterwordspacing

\end{thebibliography}
\end{document}